\documentclass[preprint]{sigplanconf}

\usepackage{url}
\usepackage{amsmath}
\usepackage{amsthm}
\usepackage{amssymb}
\usepackage[T1]{fontenc}
\usepackage[scaled]{beramono}
\usepackage{caption}
\usepackage{listings}
\usepackage{color}
\usepackage{mathtools}
\usepackage{boxedminipage}

\definecolor{eclipseGreen}{RGB}{63,127,95}

\lstdefinelanguage{lambda}{
	morekeywords={
		let, in, with, case, of, as, filter, map, query, type, if, then, else, ->, head, default, nil, cons
	},
	sensitive=false,
	morestring=[b]"
	morecomment=[l]{//}
}

\lstdefinelanguage{dl}{
	morekeywords={
	},
	sensitive=false,
	morestring=[b]",
	morecomment=[l]{//}	
}

\lstnewenvironment
	{code}[1][]{
		\lstset{
			language=lambda,
			basicstyle=\ttfamily,
			numbers=left,
			numberstyle=\tiny\ttfamily,
			keywordstyle=\bfseries,
			captionpos=b,
			frame=none,
			extendedchars=true, 
			tabsize=2,
			showtabs=false,
			columns=flexible,
			keepspaces=true,
			showstringspaces=false,
			breakatwhitespace=true,
			literate={\\}{{$\lambda$}}1 {fun}{{$\lambda$}}1 {^-}{{$^{-}$}}1 {\\top}{{$\top$}}1 {\\bot}{{$\bot$}}1 {\\exists}{{$\exists$}}1 {\\and}{{$\sqcap$}}1 {\\or}{{$\sqcup$}}1
		}
		\lstset{
			#1
		}
	}
	{}
\lstnewenvironment
	{axioms}[1][]{
		\lstset{
			language=dl,
			basicstyle=\ttfamily,
			numbers=left,
			numberstyle=\tiny\ttfamily,
			keywordstyle=\bfseries,
			commentstyle=\color{eclipseGreen},
			captionpos=b,
			frame=none,
			extendedchars=true, 
			tabsize=2,
			showtabs=false,
			columns=flexible,
			keepspaces=true,
			showstringspaces=false,
			breakatwhitespace=true,
			literate={^-}{{$^{-}$}}1 {\\top}{{$\top$}}1 {\\bot}{{$\bot$}}1 {\\forall}{{$\forall$}}1 {\\exists}{{$\exists$}}1 {\\and}{{$\sqcap$}}1 {\\or}{{$\sqcup$}}1
		}
		\lstset{
			#1
		}
	}
	{}

\theoremstyle{definition}
\newtheorem{theorem}{Theorem}
\newtheorem{lemma}{Lemma}
\theoremstyle{definition}
\newtheorem{definition}{Definition}

\theoremstyle{remark}

\newcommand\myeq{\stackrel{\mathclap{\normalfont\mbox{def}}}{=}}

\def\<#1>{\textbf{#1}\;}
\def\|#1|{\;\textbf{#1}\;}
\newcommand{\tg}[1]{\mbox{\normalfont{[}\ensuremath{\mathsf{#1}}{]}}}
\newcommand{\dnp}{\ \ \ \ \ }

\newcommand{\ir}[3]{{%
\noindent\def\arraystretch{1.5}
\ensuremath{%
\begin{array}[b]{c}
#2
\\ \hline
#3
\end{array}\hfill
\begin{tabular}[b]{r}\vspace{.5em}\tg{#1}
\end{tabular}
\medskip
}}}

\newcommand{\ax}[2]{%
\noindent
\ensuremath{%
\begin{array}{c}
#2
\end{array}\hfill
\begin{array}{r}{\tg{#1}}
\end{array}
\medskip
}}

\begin{document}

\setlength{\pdfpageheight}{\paperheight}
\setlength{\pdfpagewidth}{\paperwidth}



\begin{titlepage}
	\setlength{\parindent}{0pt}
	\rule{\textwidth}{2pt}\par
	\vspace{1.3cm}
	\sffamily{
	\Huge
	\raggedright
	\centering
	\MakeUppercase{
	$\mathbf{\lambda_{DL}}$: Syntax and Semantics}\par
	(Preliminary Report)\par
	}
	\rule{\textwidth}{1.2pt}\par
	\vspace{1.7cm}
\bfseries \MakeUppercase{Martin Leinberger}\par
University of Koblenz-Landau\par
Institute for Web Science and Technologies\par
 {mleinberger@uni-koblenz.de}\par
\bigbreak
	\bfseries \MakeUppercase{Ralf L\"{a}mmel}\par 
University of Koblenz-Landau\par
 The Software Languages Team\par
 {laemmel@uni-koblenz.de}\par
 \bigbreak
	\bfseries \MakeUppercase{Steffen Staab}\par 
University of Koblenz-Landau\par
 Institute for Web Science and Technologies\par
 {staab@uni-koblenz.de}\par
 \& Web and Internet Science Research Group\par
 University of Southampton\par
 {s.r.staab@soton.ac.uk}\par
\vfill
\medbreak
Martin Leinberger, Ralf L\"{a}mmel, Steffen Staab. $\mathbf{\lambda_{DL}}$: Syntax and Semantics (Preliminary Report).\par University of Koblenz-Landau. October 2016.\par
\end{titlepage}





\begin{abstract}
  Semantic data fuels many different applications, but is still
  lacking proper integration into programming languages. Untyped
  access is error-prone while mapping approaches cannot fully capture
  the conceptualization of semantic data. In this paper, we present
  $\lambda_{DL}$, a $\lambda$-calculus with a modified type system to
  provide type-safe integration of semantic data. This is
  achieved by the integration of description logics into the
  $\lambda$-calculus for typing and data access. It is centered around
  several key design principles. Among these are (1) the usage of
  semantic conceptualizations as types, (2) subtype inference for
  these types, and (3) type-checked query access to the data by both
  ensuring the satisfiability of queries as well as typing query results
  precisely in $\lambda_{DL}$. The paper motivates the use of a
  modified type system for semantic data and it provides the theoretic
  foundation for the integration of description logics as well as the
  core formal specifications of $\lambda_{DL}$ including a proof of
  type safety.
\end{abstract}



\keywords
Semantic data, Type systems, Typecase

\section{Introduction}
\label{sec:intro}

Semantic data allows for capturing knowledge in a natural manner. Its characteristics include the representation of conceptualizations inside the data and an entity-relation or graph-like description of data. Both, on their own and together, they allow for precisely specifying the knowledge represented within semantic data. A knowledge system manages semantic data and may infer new facts by logic inference. 
Different use cases are fueled by the semantic-data approach. The knowledge graphs of Google and Microsoft enhance Internet search. Wikidata~\cite{DBLP:journals/cacm/VrandecicK14} is an open source knowledge graph that stores structured data for Wikipedia. It consists of one billion statements and contains 1,148,230 different concepts and 2515 relations. The ontology defined by Schema.org\footnote{\url{https://schema.org/}} provides structure for data. This data is then used in search as well as personal assistants such as Google Now and Cortana. Google stores more than 3 trillion semantic statements crawled from the web. 
In the field of Life Sciences, semantic data was applied in the form of Bio2RDF\footnote{\url{http://bio2rdf.org/}}, providing 11 billion triples. Semantic data has also interlinked large, varied data sources, such as provided by Fokus\footnote{\url{https://www.fokus.fraunhofer.de/en}} containing more than 200,000 different data sets.
These examples demonstrate that semantic data models (e.g., RDF or OWL) are important for representing knowledge in complex use cases. In order to fully exploit the advantages of these data models, it is also necessary to facilitate their programmatic access and their integration into programming languages. 

As the running example, consider semantic data about music artists
formalized in the description logic \emph{ALCOI(D)}. Listing~\ref{ax:intro} shows everyone for which a $\text{\tt recorded}$ relation, that points to a entity of type $\text{\tt Song}$, exists is considered to be a $\text{\tt MusicArtist}$ (Line 2). $\text{\tt beatles}$ is of type $\text{\tt MusicArtist}$ (Line 4) and $\text{\tt machineGun}$ is a $\text{\tt Song}$ (Line 5). The object $\text{\tt hendrix}$ has recorded the song $\text{\tt machineGun}$ (Line 6) and was influenced by the object $\text{\tt beatles}$ (Line 7).
\begin{axioms}[caption={Initial example of semantic data.},label={ax:intro}]
// Conceptualization
\existsrecorded.Song $\sqsubseteq$ MusicArtist
// Graph data
beatles : MusicArtist
machineGun : Song
(hendrix, machineGun) : recorded
(hendrix,beatles) : influencedBy
\end{axioms}
The example shows several challenges we need to deal with when
integrating semantic data into a programming language. (1)
Conceptualizations rely on a mixture of nominal ($\text{\tt
  MusicArtist}$) and structural typing ($\exists\text{\tt
  recorded.Song}$). (2) It is also not uncommon to have a very general
or no conceptualization at all, as exemplified by the $\text{\tt
  influencedBy}$ role that expresses that $\text{\tt hendrix}$ has
been influenced by the $\text{\tt beatles}$. (3) Additional, implicit
statements may be derived by logical reasoning, e.g., in our running
example $\text{\tt hendrix:MusicArtist}$ can be inferred. 

Another challenge is not illustrated: (4) In real data sources, the
sheer size of potential types may become problem. It is practically
infeasible to explicitly convert all 1,148,230 different concepts of
Wikidata into types of a programming language.

Integration of data models into programming languages can be achieved
in different ways. The three most important are (1) via generic types,
(2) via a mapping to the type system of a programming language, or (3)
by using a custom type system.  A generic approach (1) can represent
semantic data using types such as \textsf{GraphNode} or \textsf{Axiom}
(cf.~\cite{DBLP:journals/semweb/HorridgeB11}). While this approach can
represent anything the data can model, it does not leverage static
typing: such generic representations are not error-checked. Mapping
approaches (2), such as~\cite{DBLP:conf/seke/KalyanpurPBP04} aim at
mapping the data model to the type system of the programming language
so that static typing is leveraged. However, the mixing of structural
and nominal typing, inferred statements, and a high number of concepts
worth mapping are problematic.

{\sloppypar

\paragraph{Contribution of the paper}

We therefore propose a third, a novel approach: A type system designed
for semantic data (3). In this paper, we present $\lambda_{DL}$, a
functional language for working with knowledge systems. $\lambda_{DL}$
uses concept expressions such as $\text{\tt MusicArtist}$ and
$\exists\text{\tt recorded.Song}$ 
as types. This ensures that every conceptualization can be represented
in the language and allows for typing values precisely. It avoids
pitfalls of other approaches by forwarding typing and subtyping
judgments to the knowledge system, thereby allowing facts to be
considered only if required. Lastly, the language contains a simple
querying mechanism based on description logics. The querying mechanism
allows for checking of satisfiability of queries as well as for typing
the query results in the programming language. As a result,
$\lambda_{DL}$ provides a type-safe method of working with semantic
data.

}

To highlight a simple kind of error that type checking can catch,
consider a function $\text{\tt f}$ that takes $\exists\text{\tt influencedBy.}\top$ as input. 
In other words, the functions accepts entities for which an $\text{\tt influencedBy}$ relation exists, irregarding of the type of entity that relation points to.
Using a query-operator that searches for entities in the data, a developer might simply query for music artists because he has seen that $\text{\tt hendrix}$ has an influence. Applying any value of the result set to the function $\text{\tt f}$ can cause runtime-errors,as not all music artists have a known influence. Typing in $\lambda_{DL}$ is precise enough to detect such errors (see Listing \ref{code:rejected}).
\begin{code}[caption={Rejected code --- music artist is not a subtype of $\exists\text{\tt influencedBy.}\top$.},label=code:rejected]
let f = fun(x:\existsinfluencedBy.\top) . x.influencedBy in
	f (head (query MusicArtist))
\end{code}

\paragraph{Road-map of the paper}

The remaining paper is organized as follows. In
Section~\ref{sec:descriptionLogics}, we introduce description logics
as the theoretic foundation of semantic data. In
Section~\ref{sec:example}, we illustrate $\lambda_{DL}$ with an
extension of the running example and an informal view on the
calculus. In Section~\ref{sec:language}, we describe the core language
and its evaluation rules. In Section~\ref{sec:typesystem}, we
describe the type system. In Section~\ref{sec:soundness}, we provide
a proof of type soundness. In Section~\ref{sec:relatedWork}, we
examine related work.  In Section~\ref{sec:conclusion}, we conclude
the paper including a discussion of future work. Additionally, we shortly describe the prototypical implementation of $\lambda_{DL}$ in the appendix. Further information about $\lambda_{DL}$ is available at \url{http://west.uni-koblenz.de/de/lambda-dl}.

 
\section{Description Logics}
\label{sec:descriptionLogics}

Semantic data is often formalized in the RDF data model or in the more expressive Web Ontology Language (OWL\footnote{\url{https://www.w3.org/OWL/}}). Formal theories about the latter are grounded in research on description logics. Description logics is a family of logical languages for describing conceptual knowledge and graph data. All description logic languages are sub-languages of first-order predicate logic. They are defined to allow for decidable or even PTIME decision procedures. 
Their usefulness for modeling semantic data has been shown with such diverse use cases as reasoning on UML class diagrams~\cite{DBLP:journals/ai/BerardiCG05}, semantic query optimization on object-oriented database systems~\cite{Beneventano:2003:DLS:762471.762472}, or improving database access through abstraction~\cite{DBLP:conf/sebd/CalvaneseGLLPR07}.



\paragraph*{Syntax and Semantics} Semantic data, also called a
knowledge base, comprises of a set of description logics axioms that
are composed using a signature $\mathit{Sig}(\mathcal{K})$ and a set
of logical and concept operators and comparisons. A signature
$\mathit{Sig}$ of a knowledge base $\mathcal{K}$ is a triple
$\mathit{Sig}(\mathcal{K})=(\mathcal{A}, \mathcal{Q}, \mathcal{O})$
where $\mathcal{A}$ is a set of concept names, $\mathcal{Q}$ is a set
of role names, and $\mathcal{O}$ is a set of object names.
DL uses Tarskian-style, interpretation-based semantics. An
interpretation $\mathcal{I}$ is a pair consisting of a non-empty
universe $\Delta^{\mathcal{I}}$ and an interpretation function
$\cdot^{\mathcal{I}}$ that maps each object $a,b\in\mathcal{O}$ to a
element of the universe. Furthermore, it assigns each concept name $A
\in \mathcal{A}$ a set $A^{\mathcal{I}} \subseteq
\Delta^{\mathcal{I}}$ and each role name $Q \in \mathcal{Q}$ to a
binary relation $Q^{\mathcal{I}} \subseteq \Delta^{\mathcal{I}} \times
\Delta^{\mathcal{I}}$. In our running example, the signature of
Listing \ref{ax:intro} contains the concepts\footnote{As is common in
  description logics research, we use ``concept C'' to refer to both the concept name $C$ and the interpretation of this concept name $C^I$, unless the distinction between the two is explicitly required. Likewise, we do for role (names) and object (names).} 
$\text{\tt MusicArtist}$ and $\text{\tt Song}$, the roles $\text{\tt recorded}$ and $\text{\tt influencedBy}$ as well as the objects $\text{\tt beatles}$, $\text{\tt hendrix}$, and $\text{\tt machineGun}$. An interpretation $\mathcal{I}$ could map objects like $\text{\tt hendrix}$ to their real-life counterparts, e.g., the artist Jimi Hendrix. 
Furthermore, the interpretation of concept $\text{\tt MusicArtist}$ might be $\text{\tt MusicArtist}^\mathcal{I} = \{ \text{\tt hendrix}, \text{\tt beatles} \}$, and the interpretation of $\text{\tt Song}$ might be $\text{\tt Song}^\mathcal{I} = \{ \text{\tt machineGun} \}$. The interpretation of the $\text{\tt recorded}$ role might be $\text{\tt recorded}^\mathcal{I} \allowbreak=\allowbreak \{ (\text{\tt hendrix},\allowbreak\text{\tt machineGun}) \}$ and $\text{\tt influencedBy}^\mathcal{I}=\{ (\text{\tt hendrix},\allowbreak\text{\tt beatles}) \}$.

Given these element names, complex expressions such as shown in Listing \ref{ax:intro} can be built. For the course of the paper, the specific description logics dialect needed to cover all necessary constructs is \emph{ALCOI}, consisting of the most commonly used \emph{Attributive Language with Complements} plus the addition of nominal concept expressions and inverse role expressions. Table \ref{tbl:roleExpressions} summarizes syntax and semantics of role expressions represented through the metavariable $R$. A role expression is either a atomic role or the inverse of a role expression.

\begin{table}[h!]
	\centering
	\begin{tabular}{l l l}
		\hline \\ [-2.0ex]
		Role Expression & Syntax & Semantics \\ \hline \hline \\ [-2.0ex]
		Atomic Role & $Q$ & $Q^{\mathcal{I}} \subseteq \Delta^{\mathcal{I}} \times \Delta^{\mathcal{I}}$ \\ \hline \\ [-2.0ex]
		Inverse & $R^{-}$ & $\{ (b,a) \in \Delta^{\mathcal{I}} \times \Delta^{\mathcal{I}} \vert (a,b) \in R^{\mathcal{I}} \}$ \\ \hline \\ [-2.0ex]
	\end{tabular}
	\caption{Role expressions and associated semantics.}
	\label{tbl:roleExpressions} 
\end{table}

Concept expressions are composed from other concept expressions and may also include role expressions. Concept expressions, represented through the metavariables $C$ and $D$, are either atomic concepts, $\top$, $\bot$ or the negation of a concept. Concept expressions can also be composed from intersection or through existential and universal quantification on a role expression. An example of such a concept expression from Listing \ref{ax:intro} is the concept $\exists\text{\tt recorded.Song}$ that describes the set of objects, which have recorded at least one song. Lastly, it is also possible to define a concept by enumerating its objects. This constitutes a nominal type in description logics and allows the description of sets such as the one only containing $\text{\tt hendrix}$ and the $\text{\tt beatles}$ through the expression $\{\text{\tt hendrix}\} \sqcup \{\text{\tt beatles}\}$. Table \ref{tbl:conceptExpressions} summarizes the syntax and semantics of concept expressions.

\begin{table}[h!]
\centering
\begin{tabular}{l l l}
	\hline \\ [-2.0ex]
	Concept Expression & Syntax & Semantics \\ \hline \hline \\ [-2.0ex]
	Nominal concept & $\{\;a\;\}$ & $\{ a^\mathcal{I} \}$ \\ \hline \\ [-2.0ex]
	Atomic concept & $A$ & $A^{\mathcal{I}} \subseteq \Delta^{\mathcal{I}}$ \\ \hline \\ [-2.0ex]
	Top & $\top$ & $\Delta^{\mathcal{I}}$ \\ \hline \\ [-2.0ex]
	Bottom & $\bot$ & $\emptyset$ \\ \hline \\ [-2.0ex]
	Negation & $\neg C$ & $\Delta^{\mathcal{I}}\setminus C$ \\ \hline \\ [-2.0ex]
	Intersection & $C \sqcap D$ & $C^\mathcal{I} \cap D^\mathcal{I}$ \\ \hline \\ [-2.0ex]
	Union & $C \sqcup D$ & $C^\mathcal{I} \cup D^\mathcal{I}$ \\ \hline \\ [-2.0ex]
	Existential Quantification & $\exists R.C$ & $\{ a^\mathcal{I} \in \Delta^\mathcal{I} \vert \exists b^\mathcal{I} : (a^\mathcal{I},b^\mathcal{I}) $\\ & &$ \in R^\mathcal{I} \wedge b^\mathcal{I} \in C^\mathcal{I} \}$ \\ \hline \\ [-2.0ex]
	Universal Quantification & $\forall R.C$ & $\{ a^\mathcal{I} \in \Delta^\mathcal{I} \vert \forall b^\mathcal{I}:(a^\mathcal{I},b^\mathcal{I}) $\\ & &$ \in R^\mathcal{I} \wedge b^\mathcal{I} \in C^\mathcal{I}\}$ \\ \hline \\	
\end{tabular}
\caption{Concept expressions and associated semantics.} 
\label{tbl:conceptExpressions}
\end{table}

Furthermore, in the context of programming with semantic data, it makes sense to add additional data types such as string or integer. We then arrive at the language \emph{ALCIO(D)}, the language \emph{ALCIO} plus the addition of data types for constructing knowledge bases. In the OWL standard, the use of XSD\footnote{\url{https://www.w3.org/TR/xmlschema-2/}} data types is common. We therefore also include XSD data types wherever it is appropriate. As an example, consider the concept expression $\exists\text{\tt artistName.xsd:string}$ describing the set of all objects having an artist name that is a string. However, as the integration of such smaller, closed set of data types can be achieved via mappings to appropriate types in the programming language, we do not go into details about them in the remainder of the paper. 

Given such concept (and datatype) expressions, we may now define semantic statements, also called a knowledge base, as pointed out before. A knowledge base $\mathcal{K}$ is a pair $\mathcal{K = (T,A)}$ consisting of the set of terminological axioms $\mathcal{T}$, the conceptualization of the data and the set of assertional axioms $\mathcal{A}$, the actual data. Schematically, a knowledge base can express that two concepts are either equivalent or that two concepts are in a subsumptive relationship. In terms of actual data, objects can either express that belong to a certain concept or that they are related to another object via a role. Furthermore, it is possible to axiomatize that two objects are equivalent. Table \ref{tbl:axiomConstruction} summarizes syntax and semantics of possible axioms in the knowledge base.

\begin{table}[h!]
\centering
\begin{tabular}{lll}
	\hline \\ [-2.0ex]
	Name & Syntax & Semantics \\ \hline \hline \\ [-2.0ex]
	Concept inclusion & $C \sqsubseteq D$ & $C^{\mathcal{I}} \subseteq D^{\mathcal{I}}$ \\ \hline \\ [-2.0ex]
	Concept equality & $C \equiv D$ & $C^{\mathcal{I}} = D^{\mathcal{I}}$ \\ \hline \hline \\ [-2.0ex]
	Concept assertion & $a : C$ & $a^\mathcal{I} \in C^\mathcal{I}$ \\ \hline \\ [-2.0ex]
	Role assertion & $(a,b):R$ & $(a^{\mathcal{I}}, b^{\mathcal{I}}) \in R^{\mathcal{I}}$ \\ \hline \\ [-2.0ex]
	Object equivalence & $a \equiv b$ & $a^{\mathcal{I}} = b^{\mathcal{I}}$ \\ \hline \\ [-2.0ex]
\end{tabular}
\caption{Terminological and assertional axioms.} 
\label{tbl:axiomConstruction}
\end{table}

Even weak axiomatizations such as RDFS\footnote{RDF Schema, one of the weakest forms of terminological axioms.} allow for the definition of domains and ranges of roles used in the ontology. As shown in Figure \ref{fig:dlAbbrev}, Domain and Range definition can be defined as abbreviations of axioms built according to Table \ref{tbl:axiomConstruction}.

\begin{figure}[h!]
	\begin{align*}
		\text{Domain}(R,C) \; & \myeq \; \exists R.\top \sqsubseteq C\\
		\text{Range(R,C)} \; & \myeq \; \top \sqsubseteq \forall R.C
	\end{align*}
\caption{Syntactical abbreviations for DL.}
\label{fig:dlAbbrev}
\end{figure}

Using our running example, we can now define a more sophisticated
knowledge base (Listing \ref{ax:example}). We assume everyone who has
recorded a song to be a music artist, but not all music artists have
recorded one (Line 2). Music artists who have been played at a radio
station however must have recorded a song (Line 3--4). Music groups
are a special kind of music artists (Line 5). Every music artist has
an artist name, which is always of type xsd:string (Line 6 and 7). As
might happen when semantic data is crawled from the Web, a role like
$\text{\tt influenceBy}$ might not be defined in the schema. Thus, it
remains a role that is not restricted by any terminological axiom. The
actual data includes descriptions of the $\text{\tt beatles}$, which
are a music group (Line 9), $\text{\tt machineGun}$ which is a song
(Line 10) $\text{\tt coolFm}$ which is a radio station (Line
11). $\text{\tt machineGun}$ has been recorded by $\text{\tt hendrix}$
(Line 12), who has been $\text{\tt influencedBy}$ the $\text{\tt
  beatles}$ (Line 13). Lastly, we know that both, $\text{\tt hendrix}$
and $\text{\tt beatles}$ have been played by $\text{\tt coolFm}$ (Line
14--15). It is not explicitly stated that $\text{\tt hendrix}$ is a
music artist. Furthermore, even though we know that the music group
$\text{\tt beatles}$ has been played at $\text{\tt coolFm}$, we do not know any song that they recorded.  

\begin{axioms}[caption={Advanced example of semantic data.},label=ax:example]
// Conceptualization
\existsrecorded.Song $\sqsubseteq$ MusicArtist
MusicArtist $\sqcap$ \existsplayedAt.RadioStation $\sqsubseteq$
	\existsrecorded.Song
MusicGroup $\sqsubseteq$ MusicArtist
MusicArtist $\sqsubseteq$ \existsartistName.\top
Range(artistName, xsd:String)
// Graph data
beatles : MusicGroup
machineGun : Song
coolFm : RadioStation
(hendrix, machineGun) : recorded
(hendrix, beatles) : influencedBy
(hendrix, coolFm) : playedAt
(beatles, coolFm) : playedAt
(hendrix, "Jimmy Hendrix") : artistName
(beatles, "The Beatles") : artistName
\end{axioms}

As illustrated by the example, ALCIO(D) is a description logics language which is already rather expressive to describe complex concept and object relationships. As we want to focus on the ``essence of programming with semantic data'', we refrain from using more powerful languages, such as OWL2DL, as this would distract from the core contributions of this paper without significantly changing its methods.

\paragraph*{Inference} 
In terms of inference, interpretations have to be reconsidered. Axioms built according to Table \ref{tbl:axiomConstruction} may or may not be true in a given interpretation. An interpretation $I$ is said to satisfy an axiom $F$, if its considered to be true in the interpretation. The notation $I\models F$ is used to indicate this. An interpretation $I$ satisfies a set of axioms $\mathcal{F}$, if $\forall F\in\mathcal{F}:I\models F$. An interpretation that satisfies a knowledge base $\mathcal{K =(T,A)}$, written $I\models\mathcal{K}$ if $I\models\mathcal{T}$ and $I\models\mathcal{A}$, is also called a model. For an axiom to be inferred from the given facts, the axiom needs to be true in all models of the knowledge base (see Def. \ref{def:inference}).
\begin{definition}[Inference]\label{def:inference}
Let $\mathcal{K=(T,A)}$ be a knowledge base, $F$ an axiom and $\mathcal{I}$ the set of all interpretations. $F$ is inferred, written $\mathcal{K}\models F$, if $\forall I\in\mathcal{I}:I\models\mathcal{K}\;\text{then}\;I\models F$.
\end{definition}
An example of this is the axiom $\text{\tt hendrix:MusicArtist}$. $\text{\tt hendrix}$ has recorded a song and must therefore be element of $\exists \text{\tt recorded.Song}$. As $\exists \text{\tt recorded.Song} \sqsubseteq \text{\tt MusicArtist}$ must be true in all models, $\text{\tt hendrix}$ must also be element of $\text{\tt MusicArtist}$. A knowledge system might introduce anonymous objects to fulfill the explicitly given axioms. Take the object $\text{\tt beatles}$ as an example. The object is a music artist and has been played in the radio. Therefore, according to Lines 3--4 in the example, they must have recorded a song. However, the knowledge system does not know any song recorded by them. It will therefore introduce an anonymous object representing this song in order to satisfy the axioms. 

\paragraph*{Queries}
Interaction between the programming language and the knowledge system can be realized via querying. Two basic forms of queries can be distinguished. Queries that check whether an axiom is true have already been introduced in the previous paragraph ($\mathcal{K}\models F)$. A more expressive form of querying introduces variables, to which the knowledge system responds with unifications for which the axiom is true. Querying introduces variables, to which the knowledge system responds with unifications for which the axiom is known to be true (see Def. \ref{def:querying}).
\begin{definition}[Querying with variables]\label{def:querying}
Let $\mathcal{K}$ be a knowledge base and $C$ a concept expression. The set of all objects for which $a:C$ is true is then $\{ \mathit{?X} \vert \mathcal{K}\models \mathit{?X}:C\}$.
\end{definition}
As an example, consider the query $\mathcal{K} \models \mathit{?X}:\text{\tt MusicArtist}$, the variable $\mathit{?X}$ is unified with all objects that belong to the concept $\text{\tt MusicArtist}$. However, this form of query can be problematic as, depending on the knowledge system, an infinite number of unifications might exist. Consider the knowledge base in Listing \ref{ax:infinite}. A person is someone who has a father who is again a person (Line 1). An object \text{\tt someone} is defined to be a person (Line 2). 
\begin{axioms}[caption={Infinitely large knowledge system.},label=ax:infinite]
Person $\sqsubseteq$ $\exists$hasFather.Person
someone : Person
\end{axioms}
If $\text{\tt someone}$ is a person, then he must have a father which is a anonymous object and a person himself, again implying that this anonymous object has a father. A query $\mathcal{K} \models \mathit{?X}:\text{\tt Person}$ therefore yields an infinite number of unifications. We therefore use a simple form of so called DL-safe queries (cf.~\cite{DBLP:journals/ws/MotikSS05}), which restrict unifications to objects defined in the signature (see Def. \ref{def:dl-safe-querying}). 
\begin{definition}[DL-safe queries]\label{def:dl-safe-querying}
Let $\mathcal{K}$ be a knowledge base, $\mathit{Sig}(\mathcal{K})=(\mathcal{A}, \mathcal{Q}, \mathcal{O})$ its signature and $C$ a concept expression. The set of all objects for which $a:C$ is true and that are not anonymous can be queried by $\{ \mathit{?X} \vert \mathcal{K}\models \mathit{?X}:C \wedge \mathit{?X}\in\mathcal{O}\}$.
\end{definition}
In this case of the example shown in Listing \ref{ax:infinite}, only the object $\text{\tt someone}$ would be returned, even though anonymous objects are considered for inferencing. 

\paragraph*{Open World and No Unique Name assumption}
Semantic data employs an open world semantics. Axioms are $\mathit{true}$ if they are true in all models of the knowledge base. Likewise, an axiom is $\mathit{false}$ if they are false in all models of the knowledge. Contrary to a closed world, axioms that are true in some models, but false in others are not false but rather $\mathit{unknown}$. This allows the modeling of incomplete data without inconsistencies. Furthermore, there is no unique name assumption. Two syntactically different objects might be equivalent. As an example, consider the two objects $\text{\tt prince}$ and $\text{\tt theArtistFormerlyKnownAsPrince}$. While they are syntactically different, they might be semantically equivalent.


\section{$\mathbf{\lambda_{DL}}$ in a nutshell}
\label{sec:example}

Developing applications for knowledge systems, as introduced in the previous section, is difficult and error-prone. $\lambda_{DL}$ has been created to achieve a type-safe way of programming with such data sources.

\subsection{Key design principles}

\paragraph*{Concepts as types} 
Type safety can only be achieved if terms are typed precisely. This is only possible if the conceptualizations of semantic data are usable in the programming language. Therefore, concept expressions must be seen as types in the language. 
\paragraph*{Subtype inferences}
Due to the large number of potential concepts, it is infeasible to compute subtype relations beforehand. Therefore, the facts about subsumptive relationships between concepts must be added to the system during the type checking process by forwarding these checks to the knowledge system.  
\paragraph*{Typing of queries}
To avoid runtime errors, queries must be properly
type-checked. Queries can be checked in two ways: First, unsatisfiable
queries must be rejected. Queries for which no possible A-Box instance can produce a result are therefore detected and rejected. 
Second, usage of queries must be type safe, meaning that the query result must be properly typed. Queries always return lists in $\lambda_{DL}$.
\paragraph*{DL-safe queries}
A knowledge system might introduce anonymous objects to satisfy axioms. In the worst case, this can lead to infinitely large query results. However, very little information can be gained of such objects aside from their existence. As shown in Def.~\ref{def:dl-safe-querying}, $\lambda_{DL}$ relies on a simplified form of DL-safe queries. Queries are enforced to be finite by only allowing unifications with known objects. However, this may also lead to empty result sets for queries. 
\paragraph*{Open-world querying}
When looking at inferencing, axioms may be $\mathit{true}$, $\mathit{false}$ or $\mathit{unkown}$. For simplicity, $\lambda_{DL}$ considers axioms to be true only if the axiom is $\mathit{true}$ in all models. In other cases, the axiom is considered false. While this view is close to a developers expectation, it also introduces the side effect that union of two queries such as $\textbf{\tt query}\;C$ and $\textbf{\tt query}\;\neg C$ does not yield all objects. For some objects, it is simply unknown whether they belong to either $C$ or $\neg C$.

\begin{figure*}[tb!]
\parbox{0.5\textwidth}{
	\begin{tabular}{l p{4.5cm}}
		$t$ ::= & \hfill ($\mathit{terms}$) \\
		& \textbf{let} $x=t$ \textbf{in} $t$ \hfill (let binding) \\
		& | \textbf{fix} $t$ \hfill (fixed point of $t$) \\
		& | $t$ $t$ \hfill (application) \\
		& | \textbf{if} $t$ \textbf{then} $t$ \textbf{else} $t$ \hfill (if-then-else) \\
		& | \textbf{cons} $t$ $t$ \hfill (list constructor) \\
		& | \textbf{null} $t$ \hfill (test for empty list) \\
		& | \textbf{head} $t$ \hfill (head of a list) \\
		& | \textbf{tail} $t$ \hfill (tail of a list) \\
		& | \textbf{query} $C$ \hfill (query) \\
		& | $t.R$ \hfill (projection) \\
		& | \textbf{case} $t$ \textbf{of} \hfill (typecase) \\
		& $\;\;\;\overline{\mathit{case}}$ \hfill (typecases) \\
		& $\;\;\;$\textbf{default} $t$ \hfill (default case) \\
		& | $t$ \textbf{=} $t$ \hfill (equivalence) \\
		& | $x$ \hfill(identifier) \\
		& | $v$ \hfill (value) \\
		& \\
		$v$ ::= & \hfill ($\mathit{values}$) \\
		& $a$ \hfill (object) \\
		& | \textbf{nil[}$T$\textbf{]} \hfill (empty list) \\
		& | \textbf{cons} $v$ $v$ \hfill (list constructor) \\
		& | $\lambda(x:T).t$ \hfill (abstraction) \\
		& | $p$ \hfill (primitive value) \\
	\end{tabular}
}
\parbox{0.5\textwidth}{
	\begin{tabular}{l p{4.5cm}}
		& \hfill \\
		$p$ ::= & \hfill ($\mathit{primitive\;values}$) \\
		& true \hfill (true) \\
		& | false \hfill (false) \\
		& \\
		$\mathit{case}$ ::= & \textbf{type} $C$ \textbf{as} $x$ -> $t$ \hfill ($\mathit{type case}$)\\
		& \\
		$T$ ::= & \hfill ($\mathit{types}$) \\
		& $C$ \hfill (concept type) \\
		& | $T \rightarrow T$ \hfill (function type) \\
		& | $T$ list \hfill (list type) \\
		& | $\Pi$ \hfill (primitive types) \\
		& \\
		$\Pi$ ::= & \hfill ($\mathit{primitive~types}$) \\
		& bool \hfill (boolean) \\
		& \\
		$\Gamma$ ::= & \hfill ($\mathit{context}$) \\
		& $\emptyset$ \hfill (empty context) \\
		& | $\Gamma,x:T$ \hfill (type binding) \\
	\end{tabular}

}
\caption{Syntax (terms, values, types) of $\lambda^{DL}$.}
\label{fig:syntax}
\end{figure*}

\begin{figure}[ht!]
\begin{align*}
\<letrec> x:T_1=t_1 \|in| t_2 \; & \myeq \; \<let> x = \|fix| (\lambda x:T_1.t_1) \|in| t_2
\end{align*}
\caption{Syntactical abbreviations of $\lambda^{DL}$.}
\label{fig:programmingAbbrev}
\end{figure}

\subsection{Example use case}

Consider an application that works on the knowledge system defined in Listing \ref{ax:example}. Four necessary functions should be implemented: First, the application should query for all music artists that have recorded a song. Second, the application should provide a mapping from a music artist to the list of their songs. Third, a mapping from a music artist to his artist name must be created. Fourth, the application should display all influences of an artist --- therefore a mapping from a music artist to his influences is needed. However, these influences should also be human-readable, meaning that they should also be mapped to their name. 

The first requirement is implemented by the querying mechanism in $\lambda_{DL}$. The necessary list of music artists that have recorded at least one song can be queried using $\text{\tt MusicArtist} \sqcap \exists\text{\tt recorded.Song}$ (see Listing \ref{code:query}). Applied to a knowledge system working on the facts in Listing \ref{ax:example}, this yields a list containing both $\text{\tt hendrix}$ and $\text{\tt beatles}$. This expression is typed by the concept expression used in the querying, assigning a type of $(\text{\tt MusicArtist} \sqcap \exists\text{\tt recorded.Song})\;\text{list}$ to the evaluation result.
\begin{code}[caption={Querying for music artists that have recorded a song.},label={code:query}]
query MusicArtist $\sqcap$ \existsrecorded.Song
\end{code}
Mapping a member of this list to his recorded songs can be done using role projections. The input type for such a mapping function is $\exists\text{\tt recorded.Song}$ which is a super type of $\text{\tt MusicArtist} \sqcap \exists\text{\tt recorded.Song}$. Listing \ref{code:mappingSongs} shows the code for the mapping function. As mentioned before, for the object $\text{\tt beatles}$, the semantic data does not contain any recorded songs, even though such a song must exist. The anonymous object introduced by the knowledge system is removed and an empty list is returned. Yet, the developer knows that an anonymous object must exist and that the knowledge system might know this song at some point in the future --- otherwise typing would have rejected the function application.  
\begin{code}[caption={Mapping to the recordings.},label={code:mappingSongs}]
let getRecordings = fun(a:\existsrecorded.Song). 
	a.recorded
\end{code}
A function mapping a music artist to his name is again built by role projections. As our knowledge systems claims that every music artist has an artist name (Listing \ref{ax:example}, line 5), the input type for this function can be the music artist concept. Additionally, the knowledge system states that the returned list of values are all of type string. We can therefore simply take the head of the returned list. Listing \ref{code:mappingName} shows the code of the mapping function. However, this code also shows a problem $\lambda_{DL}$ still faces --- if the knowledge system would not know the name of an artist, the resulting list would be empty and the code would still produce a runtime error. 
\begin{code}[caption={Mapping a artist to his name.},label={code:mappingName}]
let getArtistName = fun(a:\existsartistName.xsd:string).
	head (a.artistName)
\end{code}
The last requirement, mapping a music artist to his influences introduces casting, as music artists are not in a direct subtype relation to $\text{\tt influencedBy.}\top$. This casting is important, as simply allowing the projection could cause runtime errors if, e.g., used on the object $\text{\tt beatles}$. $\lambda_{DL}$ provides a type dispatch for this use case. Listing \ref{code:musicArtistInfluences} shows the code for this function. In case that the argument of the function is of type $\text{\tt influencedBy.}\top$, the actual mapping function is applied to the value --- otherwise, an empty list is returned. 
\begin{code}[caption={Casting a music artist to $\text{\tt influencedBy.}\top$.}, label={code:musicArtistInfluences}]
let getArtistInfluences = fun(artist:MusicArtist).
	case artist of
		type \existsinfluencedBy.\top as x -> getInfluences x
		default nil
\end{code} 
The function computing the actual influences can use a projection and then a mapping to a human-readable name. However, this human-readable name is problematic. Due to the weak schematic restrictions of the $\text{\tt influencedBy}$ role
the code must proceed on a case by case basis. If the influence is a music artist, the projection to the human-readable string is known. Otherwise, the influence should be ignored. Listing \ref{code:getInfluences} shows the complete code for the function.    
\begin{code}[caption={Mapping influences to their human-readable representations.},label={code:getInfluences}]
let getInfluences = fun(obj:\existsinfluencedBy.\top).
	let toName = fun(x:\existsinfluencedBy^-.\top).
		case x of
			type MusicArtist as y -> getName y
			default "no influence known"
	in letrec getNames:(\existsinfluencedBy^-.\top list 
		-> string list) =
		fun(source:\existsinfluencedBy^-.\top list) . 
			if (null source) 
				then nil[string]
				else cons (toName (head source)) 
					(getNames (tail source))
	in
		getNames obj.influencedBy
\end{code}


\begin{figure}[tb!]
\centering
\begin{boxedminipage}{0.5\textwidth}
	\ax{\text{E-QUERY}}
		{\<query> C \rightarrow \sigma(\{\mathit{?X}\;\vert\;\mathit{?X}\in\mathcal{O}\wedge\mathcal{K} \models \mathit{?X}:C\})\\=(\textbf{cons}\;a_1\;...)}

	\ax{\text{E-PROJV}}
		{a.R \rightarrow \sigma(\{\mathit{?X}\;\vert\;\mathit{?X}\in\mathcal{O}\wedge\mathcal{K} \models (a,\mathit{?X}):R\})\\=(\textbf{cons}\;b_1\;...)}

	\ir{\text{E-PROJ}}
		{t_1 \rightarrow t'_1}
		{t_1.R \rightarrow t'_1.R}

	\ir{\text{EQ-NOMINAL-TRUE}}
		{\mathcal{K} \models a \equiv b}
		{a\textbf{=}b\rightarrow \text{true}}

	\ir{\text{EQ-NOMINAL-FALSE}}
		{\mathcal{K} \not\models a \equiv b}
		{a \textbf{=} b \rightarrow \text{false}}

	\ax{\text{EQ-PRIM-TRUE}}
		{p_1 \textbf{=} p_1 \rightarrow \text{true}}

	\ir{\text{EQ-PRIM-FALSE}}
		{p_1 \neq p_2}
		{p_1 \textbf{=} p_2 \rightarrow \text{false}}

	\ir{\text{E-EQ1}}
		{t_1 \rightarrow t'_1}
		{t_1 = t_2 \rightarrow t'_1 = t_2}

	\ir{\text{E-EQ2}}
		{t_2 \rightarrow t'_2}
		{v_1 = t_2 \rightarrow v_1 = t'_2}
\end{boxedminipage}
\caption{Reduction rules related to KB.}
\label{fig:kbSemantics}
\end{figure}

\begin{figure*}[tb!]
\centering
\begin{boxedminipage}{0.7\textwidth}
\ax{\text{E-DISPATCH-DEF}}
	{\<case> a \|of| \<default> t_0 \rightarrow t_0}
\noindent\def\arraystretch{1.5}
	\ensuremath{%
	\begin{array}[b]{lcr}
	& \mathcal{K}\models a:C_1 &
	\\ \hline
	\<case> a \|of| & & \\ 
	\;\;\;\<type> C_1 \|as| x_1 \|->| t_1 & & \\
	\;\;\;... & \rightarrow & \lbrack x_1 \mapsto a \rbrack t_1\\
	\;\;\;\<default> t_{n+1} & &  \\
	\end{array}\hfill
	\begin{tabular}[b]{r}\vspace{5.5em}\tg{\text{E-DISPATCH-SUCC}}
	\end{tabular}
	\medskip
	}
\noindent\def\arraystretch{1.5}
	\ensuremath{%
	\begin{array}[b]{lcl}
	& \mathcal{K}\not\models a:C_1 &
	\\ \hline
	\<case> a \|of| & & \<case> a \|of| \\ 
	\;\;\;\<type> C_1 \|as| x_1 \|->| t_1 & &\;\;\;\<type> C_2 \|as| x_2 \|->| t_2 \\
	\;\;\;\<type> C_2 \|as| x_2 \|->| t_2 & \rightarrow & \;\;\;... \\ 
	\;\;\;... &  &  \;\;\;\<default> t_{n+1} \\
	\;\;\;\<default> t_{n+1} & &  \\
	\end{array}\hfill
	\begin{tabular}[b]{r}\vspace{7.0em}\tg{\text{E-DISPATCH-FAIL}}
	\end{tabular}
	\medskip
	}
\noindent\def\arraystretch{1.5}
	\ensuremath{%
	\begin{array}[b]{lcl}
	t_1 & \rightarrow & t'_1
	\\ \hline
	\<case> t_1 \|of| & & \<case> t'_1 \|of| \\ 
	\;\;\; \overline{\mathit{case}} &  \rightarrow &\;\;\; \overline{\mathit{case}} \\
	\;\;\;\<default> t_{n+1} & & \;\;\;\<default> t_{n+1} \\
	\end{array}\hfill
	\begin{tabular}[b]{r}\vspace{4.0em}\tg{\text{E-DISPATCH}}
	\end{tabular}
	\medskip
	}
\end{boxedminipage}
\caption{Reduction rules for type case terms.}
\label{fig:evaluatingTypeCase}
\end{figure*}

\section{Core language}
\label{sec:language}

\paragraph*{Syntax}
Our core language (Figure \ref{fig:syntax}) is a simply typed call-by-value $\lambda$-calculus. Terms of the language include let-statements, a fixed point operator for recursion, function application and if-then-else statements. Constructs for lists are included in the language: cons, nil, null, head and tail. Based on these, complex operations such as map, fold and filter can be built. For simplicity, we did not include these in our syntax. Specific to our language is the querying construct for selecting data in the knowledge system based on a concept expression and projections from an object to a set of objects using role expressions. Casting is done via a type-dispatch construct that contains an arbitrary number of cases plus a default case. We use an overbar notation to represent sequences of syntactical elements. Concretely, $\overline{a}$ stands for $a_1, a_2, ..., a_n$. As DL has no unique name assumption, objects can be syntactically different but semantically equivalent. Therefore, we also included the equality operator in our representation.    
Values ($v$) include objects defined in the knowledge base, nil and cons to represent lists, $\lambda$-abstractions and primitive values. $\lambda$-abstractions indicate the type of their variable. In terms of primitive values, we assume data types such as integers and strings, but omit routine details. To illustrate them, we usually just include booleans in our syntax. Types ($T$) consist of concept expressions built according to Table \ref{tbl:axiomConstruction}, type constructors for function and list types and primitive types. Additionally, we use a typing context to store type bindings for $\lambda$-abstractions.
To simplify recursion, we also define a letrec as an abbreviation of the fixpoint operator (see Figure \ref{fig:programmingAbbrev}).

\paragraph*{Semantics}
The operational semantics is defined using a reduction relation, which extends the standard approaches. Reduction of lists and terms not related to the knowledge bears no significant difference from rules as, e.g., defined in \cite{Pierce:2002:TPL:509043}. We therefore show these rules in the appendix and focus on the terms related to the knowledge base (see Figure \ref{fig:kbSemantics}). 

A term representing a query can be directly evaluated to a list of objects (E-QUERY). The query reduction rule queries the knowledge system for all $\mathit{?X}$ for which the axiom $\mathcal{K}\models \mathit{?X}:C$ is true. As $\lambda_{DL}$ relies on DL-safe queries, only objects actually defined in the signature are allowed. For simplicities sake, we consider the result to be a list and introduce a $\sigma$-operator that takes care of communication between the knowledge system and $\lambda_{DL}$. Projections (E-PROJ and E-PROJV) behave similarly. Once the term has been reduced to a object $a$, the knowledge system is queried for all $\mathit{?X}$ for which $\mathcal{K}\models (a,\mathit{?X}):R$. Again, anonymous objects are not considered and the result is converted into a list by the $\sigma$-operator. In case of equivalence, both terms must first be reduced to values (E-EQ1 and E-EQ2). Once both terms are values, equivalence can be computed. Equivalence is distinguished into equivalence for objects (EQ-NOMINAL-TRUE and EQ-NOMINAL-FALSE) and equivalence for primitive values (EQ-PRIM-TRUE and EQ-PRIM-FALSE). $\lambda_{DL}$ considers two primitive values only equivalent if they are syntactically equal. In case of objects, the knowledge base is queried. If the knowledge system can unambiguously prove that $a$ is equivalent to $b$, the two objects are considered to be equal. Due to the open-world querying, objects are considered to be different if the knowledge system is unsure or if it can actually prove that the two objects are not equivalent. We do not consider equivalence for lists or $\lambda$-abstractions and avoid these cases during type-checking.

Evaluation of type-dispatch terms (see Fig. \ref{fig:evaluatingTypeCase}) is somewhat special. The terms to be dispatched is first reduced to a object (E-DISPATCH). The semantics can then test the object case by case until one of them matches (E-DISPATCH-SUCC and E-DISPATCH-FAIL). For each case the knowledge system is queried whether the axiom $\mathcal{K}\models a:C$ is true. Due to the open-world querying, it might happen that the knowledge system cannot compute such a membership. In this case, the type-dispatch uses its default case to continue evaluation.



\section{Type system}
\label{sec:typesystem}

The most distinguishing feature of the type system for $\lambda_{DL}$ is the addition of concept expressions, built according to the rules described in Table \ref{tbl:conceptExpressions}, as types in the language. For constructs unrelated to the knowledge system, this has little impact. However, computation of upper and lower bounds change due to the addition of concepts. 

\paragraph*{Least-Upper Bound and Greatest-Lower Bound}
Computation of the least-upper bound of two types $S$ and $T$, as, e.g., required for typing if-then-else terms is done by a special judgment dubbed $lub$ (see Fig. \ref{fig:lup}). In case of a least-upper bound for primitive types, we simply assume the types to be equal, making the least-upper bound the type itself (LUB-PRIMITIVE). For two concepts $C$ and $D$, a new concept $C \sqcup D$ is constructed (LUB-CONCEPT). For lists of the form $S\;\text{list}$ and $T\;\text{list}$, we compute the least-upper bound of $S$ and $T$ as a new type for the list. For two functions, $S_1 \rightarrow S_2$ and $T_1 \rightarrow T_2$, the greatest-lower bound of the types $S_1$ and $T_1$ as well as the least upper bound of $S_2$ and $T_2$ are computed. 

\begin{figure}[h!]
\centering
\begin{boxedminipage}{0.45\textwidth}
	\ax{\text{LUB-PRIMITIVE}}
		{lub(\pi_1,\pi_1) \Rightarrow \pi_1}

	\ax{\text{LUB-CONCEPT}}
		{lub(C,D)\Rightarrow C \sqcup D}

	\ir{\text{LUB-LIST}}
		{lub(S,T) \Rightarrow W}
		{lub(S\;\text{list},T\;\text{list}) \Rightarrow W\;\text{list}}

	\ir{\text{LUB-FUNC}}
		{glb(S_1,T_1)\Rightarrow W_1 \dnp lub(S_2, T_2) \Rightarrow W_2}
		{lub(S_1 \rightarrow S_2, T_1 \rightarrow T_2) \Rightarrow W_1 \rightarrow W_2}
\end{boxedminipage}
\caption{Least-upper bound of types.}
\label{fig:lup}
\end{figure}

The greatest-lower bound of two types $S$ and $T$ works analogous to the least-upper bound. Two primitive types must be again equal, making their greatest lower bound the type again. The greatest lower bound of two concepts $C$ and $D$ is the concept $C \sqcap D$. Lists are again reduced to a greatest lower bound of their type. The same is true for functions. The exact rules can be seen in the Fig. \ref{fig:glb} in the appendix.

\paragraph*{Typing knowledge-base unrelated constructs}

Given the judgment for the least upper bound of two types, the typing rules can now be defined (see Fig. \ref{fig:unrelatedTyping}). Typing of let, fixpoint operations, applications, abstractions, variables and primitive values does not differ from standard approaches. Typing of if-then-else statements relies on the $lub$-judgment to create a type $W$ that combines both branches.  

\begin{figure}[h!]
\centering
\begin{boxedminipage}{0.45\textwidth}
\ir{\text{T-LET}}
	{\Gamma \vdash t_1 : T_1 \dnp \Gamma,x:T_1 \vdash t_2 : T_2}
	{\Gamma \vdash \<let> x=t_1 \|in| t_2 : T_2}

\ir{\text{T-FIX}}
	{\Gamma\vdash t_1 : T_1 \rightarrow T_1}
	{\Gamma\vdash \<fix> t_1 : T_1}

\ir{\text{T-APP}}
	{\Gamma \vdash t_1 : T_1 \rightarrow T_2 \dnp \Gamma \vdash t_2 : T_1}
	{\Gamma \vdash t_1 t_2 : T_2}

\ir{\text{T-IF}}
	{\Gamma \vdash t_1 : \text{bool} \dnp \Gamma \vdash t_2 : S \dnp \Gamma \vdash t_3 : T \\ lub(S,T) \Rightarrow W}
	{\Gamma \vdash \<if> t_1 \|then| t_2 \|else| t_3 : W}

\ir{\text{T-ABS}}
	{\Gamma,x:T_1\vdash t_2 : T_2}
	{\Gamma \vdash \lambda (x:T_1).t_2 : T_1 \rightarrow T_2}

\ir{\text{T-VAR}}
	{x : T \in \Gamma}
	{\Gamma \vdash x : T}

\ax{\text{T-TRUE}}
	{\Gamma \vdash \text{true} : \text{bool}}

\ax{\text{T-FALSE}}
	{\Gamma \vdash \text{false} : \text{bool}}

\ir{\text{T-SUB}}
	{\Gamma \vdash t : S \dnp S <: T}
	{\Gamma \vdash t : T}
\end{boxedminipage}
\caption{Typing rules for constructs unrelated to the KB.}
\label{fig:unrelatedTyping}
\end{figure}

In terms of lists, we restrict ourselves to lists of objects for demonstration purposes. An empty list (T-NIL) can be typed using the type annotation. A cons function (T-CONS) can be typed if it is applied to a term of type $T_1$ and a term of type $T_2\;\text{list}$. The new list can be typed using the least-upper-bound judgment to create the type $T_3\;\text{list}$. The remainder are standard list typing rules: A null function takes a well-typed list and returns a boolean value. Head needs a well-typed list of type $T\;\text{list}$ and returns a value of type $T$. Tail again takes a well-typed list of type $T\;\text{list}$ and returns a list of the same type. Fig. \ref{fig:listTyping} summarizes the rules.

\begin{figure}[bt!]
\begin{boxedminipage}{0.45\textwidth}
\ax{\text{T-NIL}}
	{\Gamma \vdash \<nil[T]> : T\;\text{list}}
	
\ir{\text{T-CONS}}
	{\Gamma \vdash t_1 : T_1 \dnp \Gamma \vdash t_2 : T_2\;\text{list} \\ lub(T_1,T_2)\rightarrow T_3}
	{\Gamma \vdash \<cons> t_1 \; t_2 : T_3\;\text{list}}

\ir{\text{T-NULL}}
	{\Gamma \vdash t_1 : T\;\text{list}}
	{\Gamma \vdash \<null>t_1 : \text{Bool}}

\ir{\text{T-HEAD}}
	{\Gamma \vdash t_1 : T\;\text{list}}
	{\Gamma \vdash \<head>t_1 : T}

\ir{\text{T-TAIL}}
	{\Gamma \vdash t_1 : T\;\text{list}}
	{\Gamma \vdash \<tail>t_1 : T\;\text{list}}
\end{boxedminipage}
\caption{Typing rules for lists}
\label{fig:listTyping}
\end{figure}

\paragraph*{Typing of knowledge-base related constructs}

\begin{figure}[bt!]
\begin{boxedminipage}{0.45\textwidth}
\ir{\text{T-QUERY}}
	{\mathcal{K}\not\models C\equiv\bot}
	{\Gamma\vdash \<query> C : C\;\text{list}}
	
\ir{\text{T-PROJ}}
	{\Gamma \vdash t_1 : C}
	{\Gamma \vdash t_1.R : (\exists R^-.C)\;\text{list}}

\ir{\text{T-EQN}}
	{\Gamma \vdash t_1 : C \dnp \Gamma \vdash t_2 : D\dnp \mathcal{K}\not\models C\sqcap D \equiv\bot}
	{\Gamma\vdash t_1 \|=| t_2 : \text{bool}}

\ir{\text{T-EQP}}
	{\Gamma \vdash t_1 : \Pi_1 \dnp \Gamma \vdash t_2 : \Pi_1}
	{\Gamma\vdash t_1 \|=| t_2 : \text{bool}}

\ax{\text{T-OBJECT}}
	{\Gamma\vdash a:\{\;a\;\}}
\end{boxedminipage}
\caption{Typing rules for constructs related to the KB.}
\label{fig:kbTyping}
\end{figure}

\begin{figure}[b!]
\begin{boxedminipage}{0.45\textwidth}
\ax{S-RELF}
	{S <: S}

\ir{S-CONCEPT}
	{\mathcal{K} \models C \sqsubseteq D}
	{C <: D}

\ir{S-LIST}
	{S <: T}
	{S\;\text{list} <: T\;\text{list}}

\ir{S-FUNC}
	{T_1 <: S_1 \dnp S_2 <: T_2}
	{S_1 \rightarrow S_2 <: T_1 \rightarrow T_2}
\end{boxedminipage}
\caption{Subtyping rules.}
\label{fig:subtyping}
\end{figure}

\begin{figure*}[tb!]
\centering
\begin{boxedminipage}{0.7\textwidth}
\noindent\def\arraystretch{1.5}
	\ensuremath{%
	\begin{array}[b]{lcr}
	\Gamma\vdash t_0 : D & & \Gamma,x_i:C_i\vdash t_i : T_i\;\text{for i=}1,..n	 \\
	\mathcal{K}\not\models C_i \sqsubseteq C_j\;\text{for}\;i<j & & \mathcal{K}\not\models C_i \sqcap D \equiv \bot\;\text{for i}=1,..,n \\
	\Gamma\vdash t_{n+1} : T_{n+1} & & \overline{\mathit{lub}}(T_1,...,T_{n+1})\Rightarrow W
	\\ \hline
	\Gamma\vdash \<case> t_0 \|of| & & \\ 
	\;\;\;\<type> C_1 \|as| x_1 \|->| t_1 & & \\
	\;\;\;... & : & W\hspace{2cm}\\
	\;\;\;\<type> C_n \|as| x_n \|->| t_n & &  \\
	\;\;\;\<default> t_{n+1} & &  \\
	\end{array}\hfill
	\begin{tabular}[b]{r}\vspace{7.0em}\tg{\text{T-DISPATCH}}
	\end{tabular}
	\medskip
	}
\end{boxedminipage}
\caption{Typing rule for type case}
\label{fig:typingTypecase}
\end{figure*}

Typing of terms related to the knowledge base is summarized in Figure \ref{fig:kbTyping}. Queries (T-QUERY) have a concept associated with them - therefore, the result of the evaluation will be of type $C\;\text{list}$. To avoid unsatisfiable queries, the knowledge system is queried whether the concept $C$ is satisfiable. If it is, typing does not assign a type to the term and type-checking aborts with an error. Projections (T-PROJ) require a term of type $C$ and can then be typed by the inverse of the relation used for the projection. While this may seem confusing on first sight, it is actually the most precise type that can be assigned to this term. Range-definitions of roles are often extremely general (e.g., the range definition for $\text{\tt influencedBy}$). Equivalence (T-EQN and T-EQ-P) simply requires two well typed values that are either primitives or objects and can then be typed as $\text{bool}$ . Lastly, single objects can be typed using a nominal concept --- a concept expression created through enumerating its members. 

Typing of a type-dispatch (see Fig. \ref{fig:typingTypecase}) is similar to typing of a if-then-else. Given that the term being dispatched is a well typed concept $D$, the type of the term is the least-upper-bound of all branches. We use $\overline{lub}$ as a shortcut for the repeated application of the $lub$-judgment. Additional checks ensure meaningful cases. First of all, the intersection between $C_i$ and $D$ should not be equivalent to $\bot$, as it would then be impossible for the case to ever match. Second, since cases are checked sequentially, it should not happen that a case is subsumed by a case occurring before it. 

\paragraph*{Subtyping}
Subtyping rules are summarized in Fig. \ref{fig:subtyping}. Any type is always a subtype of itself (S-RELF). Subtyping for concepts is handled by the knowledge system. A concept $C$ is a subtype of concept $D$ if the knowledge base can infer that $\mathcal{K}\models C \sqsubseteq D$ (S-CONCEPT). The forwarding of this decision to the knowledge system is important because the knowledge system can take inferred facts into account before making the conclusion. Subtyping for list and function types is reduced to subtyping checks for their associated types. A list $S\;\text{list}$ is a subtype of $T\;\text{list}$ if $S <: T$ is true (S-LIST). A function is subsumed by another if its domain is more specific, but its co-domain more general (S-FUNC).

\paragraph*{Algorithmic type-checking}

Algorithmic type-checking is completely syntax driven. For instance,
transitivity, which could fail to be syntax-driven, is handled by the
knowledge system in case of concept expressions, while primitive types
do not include any subtype relations. 


\section{Type soundness}
\label{sec:soundness}

In this section, we prove the soundness of $\lambda_{DL}$: If a program is well-typed, it does not get stuck. As with many other languages, there are exceptions to this rule (e.g., down-casting in object-oriented languages, cf.~\cite{DBLP:journals/toplas/IgarashiPW01}). For $\lambda_{DL}$, these exceptions concern lists. We therefore show that if a program is well-typed, then the only way it can get stuck is if it reaches a point where it tries to compute $\textbf{head}\;\textbf{nil}$ or $\textbf{tail}\;\textbf{nil}$. 
We proceed in two steps, by showing that a well-typed term is either a value or it can take a step (progress) and by showing that if that term takes a step, the result is also well-typed (preservation). We start by providing some forms about the possible well-typed values (canonical forms) for each type.

\begin{lemma}[Canonical Forms Lemma]\label{lemma:canonical}
Let $v$ be a well-typed value. Then the following observations can be made:
\begin{enumerate}
\item If $v$ is a value of type $C$, then $v$ is of the form $a$.
\item If $v$ is a value of type $T_1 \rightarrow T_2$, then $v$ is of the form $\lambda(x:S_1).t_2$ with $S_1 <: T_1$.
\item If $v$ is a value of type $C\;\text{list}$, then $v$ is either of the form  $(\textbf{cons}\;v_1 ...)$ or $\textbf{nil}$.
\item If $v$ is a value of type $\text{bool}$, then either $v$ is either $\text{true}$ or $\text{false}$.
\end{enumerate}
\end{lemma}

\begin{proof}
Immediate from the typing relation. 
\end{proof}

Given Lemma~\ref{lemma:canonical}, we can show show that a well-typed term is either a value or it can take a step.

\begin{theorem}[Progress]\label{theorem:progress}
Let $t$ be a well-typed closed term. If $t$ is not a value, then there exists a term $t'$ such that $t \rightarrow t'$.
If $\Gamma \vdash t : T$, then $t$ is either a value, a term containing the forms $\textbf{head}\;\textbf{nil}$ and $\textbf{tail}\;\textbf{nil}$, or there is some $t'$ with $t \rightarrow t'$.
\end{theorem} 

\begin{proof}
By induction on the derivation of $\Gamma \vdash t : T$. We proceed by examining each case individually. 
\begin{description}

\item[(\underline{T-LET})] $t = \textbf{let}\;x\;\text{=}\;t_1\;\textbf{in}\;t_2,\;\Gamma \vdash t_1:T_1, \;\Gamma,x:T_1\vdash t_2 : T_2$. By hypothesis, $t_1$ is either a value or it can make a step. If it can, rule E-LET applies. If its a value, E-LETV applies (see Fig. \ref{fig:normalSemantics}). 

\item[(\underline{T-FIX})] $t = \textbf{fix}\;t_1,\;\Gamma\vdash t_1:T_1\rightarrow T_1,\;\Gamma\vdash t:T_1$. By induction hypothesis, $t_1$ is either a value or it can take a step. If it can take a step, rule E-FIX applies. If its a value, by the canonical forms lemma (Lemma~\ref{lemma:canonical}), $t_1 = \lambda(x:T_1).t_2$. Therefore, rule E-FIXV applies. 

\item[(\underline{T-APP})] $t=t_1\;t_2,\; \Gamma \vdash t_1:T_{11} \rightarrow T_{12}\;, \Gamma \vdash t_2:T_{11},\;\Gamma\vdash t:T_{12}$. By hypothesis, $t_1$ and $t_2$ are either a values or they can take a step. If they can take a step, rules E-APP1 or E-APP2 apply. If both are values, then by the canonical forms lemma (Lemma~\ref{lemma:canonical}), $t_1 = \lambda(x:T_{11}).t_{11}$ and rule E-APPABS applies. 

\item[(\underline{T-IF})] $t=\textbf{if}\;t_1\;\textbf{then}\;t_2\;\textbf{else}\;t_3,\;\Gamma\vdash t_1:\text{bool},\;\Gamma\vdash t_2:S,\;\Gamma\vdash t_3:T,\;\Gamma\vdash t:W$. By induction hypothesis, $t_1$ is a value or it can take a step. If it can take a step, rule E-IF applies. If its a value, then by Lemma~\ref{lemma:canonical}, $t_1 = \text{true}$ or $t_1 = \text{false}$. In this case, either rules E-IF-TRUE or E-IF-FALSE apply.

\item[(\underline{T-ABS})] Immediate since $\lambda(x:T_1).t_2$ is a value.  

\item[(\underline{T-VAR})] Impossible since we're only looking at closed terms. 

\item[(\underline{T-TRUE})] Immediate, since $\text{true}$ is a value. 

\item[(\underline{T-FALSE})] Immediate, since $\text{false}$ is a value. 

\item[(\underline{T-SUB})] Result follows from induction hypothesis. 

\item[(\underline{T-NIL})] Immediate, since $\text{nil}$ is a value.

\item[(\underline{T-CONS})] $t=\textbf{cons}\;t_1\;t_2,\;\Gamma\vdash t_1:C,\;,\Gamma\vdash t_2:D\;\text{list}$. By hypothesis, $t_1$ and $t_2$ are either values or they can take a step. If they can take a step, rules E-CONS1 and E-CONS2 apply (see Fig. \ref{fig:listSemantics}). Otherwise, the term is a value. 

\item[(\underline{T-NULL})] $t=\textbf{null}\;t_1, \;\Gamma \vdash t_1 : T\;\text{list},\;\Gamma\vdash t:\text{bool}$. By hypothesis, $t_1$ is either a value or it can take a step. If it can take a step, rule E-NULL applies. If its a value, by Lemma~\ref{lemma:canonical}, $t_1 = \textbf{nil}$ or $t_1 = (\textbf{cons}\;v_1 ...)$. Then either E-NULL-TRUE or E-NULL-FALSE apply. 

\item[(\underline{T-HEAD})] $t=\textbf{head}\;t_1,\;\Gamma\vdash t_1 : T\;\text{list},\;\Gamma\vdash t:T$. By hypothesis, $t_1$ is either a value or it can take a step. If it can take a step, rule E-HEAD applies. Otherwise, by Lemma~\ref{lemma:canonical}, $t_1 = \textbf{nil}$ or $t_1 = (\textbf{cons}\;v_1 ...)$. Then either rule E-HEADV applies or the term is in the accepted normal form $t=\textbf{head}\;\textbf{nil}$.

\item[(\underline{T-TAIL})] $t=\textbf{tail}\;t_1,\;\Gamma\vdash t_1 : T\;\text{list},\;\Gamma\vdash t:T\;\text{list}$. By hypothesis, $t_1$ is either a value or it can take a step. If it can take a step, rule E-TAIL applies. Otherwise, by Lemma~\ref{lemma:canonical}, $t_1 = \textbf{nil}$ or $t_1 = (\textbf{cons}\;v_1 ...)$. Then either rule E-TAILV applies or the term is in the accepted normal form $t=\textbf{tail}\;\textbf{nil}$.

\item[(\underline{T-QUERY})] $t=\textbf{query}\;C,\;\Gamma\vdash t:C\;\text{list}$. Immediate since rule E-QUERY applies (see Fig. \ref{fig:kbSemantics}).

\item[(\underline{T-PROJ})] $t = t_1.R,\;\Gamma\vdash t_1 : C,\;\Gamma\vdash t:(\exists R^-.C)$. By hypothesis, either $t_1$ is a value or it can take a step. If it can take a step, rule E-PROJ applies. If its a value, then by Lemma~\ref{lemma:canonical} $t_1 = a$, therefore rule E-PROJV applies.

\item[(\underline{T-DISPATCH})] $\;$
\begin{lstlisting}[aboveskip=0pt,belowskip=0pt]
$t=\textbf{case}\;t_0\;\textbf{of}$
	$\overline{\mathit{case}}$
	$\textbf{default}\;t_{n+1}$
$\Gamma\vdash t_0:D,\;\Gamma\vdash t:W$
\end{lstlisting}
By hypothesis, $t_0$ is either a value or it can take a step. If it can take a step, rule E-DISPATCH applies. If its a value, by Lemma~\ref{lemma:canonical}, $t_0 = a$. If $\overline{\mathit{case}}$ is non-empty, either rules E-DISPATCH-SUCC or E-DISPATCH-FAIL apply. Otherwise, rule E-DISPATCH-DEF applies (see Fig. \ref{fig:evaluatingTypeCase}). 

\item[(\underline{T-EQN})] $t_1\;\textbf{=}\;t_2,\;\Gamma\vdash t_1:C,\;\Gamma\vdash t_2:D$. Either $t_1$ and $t_2$ are values or they can take a step. If they can take a step, rules E-EQ1 and E-EQ2 apply. If both are values, by Lemma~\ref{lemma:canonical}, $t_1=a$, $t_2=b$. Therefore, either rule EQ-NOMINAL-TRUE or EQ-NOMINAL-FALSE applies. 

\item[(\underline{T-EQP})] $t_1\;\textbf{=}\;t_2,\;\Gamma\vdash t_1:\Pi_1,\;\Gamma\vdash t_2:\Pi_1$. Either $t_1$ and $t_2$ are values or they can take a step. If they can take a step, rules E-EQ1 and E-EQ2 apply. If both are values, them they are either syntactically equal or not. Therefore either EQ-PRIM-TRUE or EQ-PRIM-FALSE applies. 

\item[(\underline{T-OBJ})] Immediate, since $t=a$ is a value. 
\end{description}
\end{proof}

For proving preservation, two additional Lemmas are required. One, that substitution preserves the type and two, that the least-upper bound judgment computes a type that is really a supertype of its two input types. 

\begin{lemma}[Substitution]\label{lemma:substitution}
If $\Gamma,x:S\vdash t:T$ and  $\Gamma\vdash s:S$, then $\Gamma\vdash\lbrack x\mapsto s\rbrack t : T$.
\end{lemma}

\begin{proof}
Substitution in $\lambda_{DL}$ does not differ from standard approaches, e.g., as described in \cite{Pierce:2002:TPL:509043}. Therefore, the proof is omitted.
\end{proof}

\begin{lemma}[Least-Upper-Bound]\label{lemma:lub}
Let $S$, $T$ and $W$  be types. If $lub(S,T) \Rightarrow W$, then $S <: W$ and $T <: W$. 
\end{lemma}

\begin{proof}
Four cases must be considered: $S$ and $T$ are either primitives, concepts, lists or functions.
\begin{description}
	\item[Primitives:] Result is immediate since $S=T=W$. By subtyping rule S-REFL, $S <: W$ and $T <: W$ holds. 
	\item[Concepts:] $S=C,\;T=D,\;W=C\sqcup D$. Since $\mathcal{K}\models C\sqsubseteq C\sqcup D$ and $\mathcal{K}\models D\sqsubseteq C\sqcup D$, $S <: W$ and $T <: W$ hold via subtyping rule S-CONCEPT.
	\item[Lists] Immediate through the induction hypothesis and subtyping rules for lists. 
	\item[Functions] Immediate through induction hypothesis and subtyping rules for functions. 
\end{description}
\end{proof} 

Given these Lemmas, we can now continue to show that if a term takes a step by the evaluation rules, its type is preserved. 

\begin{theorem}[\textbf{Preservation}]\label{theorem:preservation}
Let $t$ be a term and $T$ a type. If $\Gamma\vdash t:T$ and $t\rightarrow t'$, then $\Gamma\vdash t':T$. 
\end{theorem}

\begin{proof}
By induction on the derivation of $\Gamma\vdash t:T$. We proceed by examining each case individually. 
\begin{description}

\item[(\underline{T-LET})] $t = \textbf{let}\;x\;\text{=}\;t_1\;\textbf{in}\;t_2,\;\Gamma\vdash t:T_2,\;\Gamma \vdash t_1:T_1, \;\Gamma,x:T_1\vdash t_2 : T_2$. There are two ways $t$ can be reduced: E-LET and E-LETV. 
\begin{description}
	\item[1] $t'=\textbf{let}\;x=t'_1\;\textbf{in}\;t_2$ By induction hypothesis, $t_1 \rightarrow t'_1$ preserves the type. Therefore, by rule T-LET, $t' : T_2$.
	\item[2] $t'=\lbrack x \mapsto v_1\rbrack t_2$. By Lemma~\ref{lemma:substitution} typing is preserved, therefore $t' : T_2$.
\end{description}

\item[(\underline{T-FIX})] $t = \textbf{fix}\;t_1,\;\Gamma\vdash t_1:T_1\rightarrow T_1,\;\Gamma\vdash t:T_1$. There are two rules by which $t$ can be reduced: E-FIX and E-FIXV.
\begin{description}
	\item[(1)] $t'=\textbf{fix}\;t'_1$. By induction hypothesis, $t_1 \rightarrow t'_1$ preserves the type. Therefore, by T-FIX, $t':T_1$.
	\item[(2)] $t'=\lbrack x \mapsto \textbf{fix}\;(\lambda (x:T_1).t_2)\rbrack t_2$. By Lemma~\ref{lemma:substitution}, substitution preserves the type. Therefore, $t':T_1$.
\end{description}

\item[(\underline{T-APP})] $t=t_1\;t_2,\; \Gamma \vdash t_1:T_{11} \rightarrow T_{12}\;, \Gamma \vdash t_2:T_{11},\;\Gamma\vdash t:T_{12}$. There are three rules by which $t'$ can be computed: E-APP1, E-APP2 and E-APPABS.
\begin{description}
	\item[(1)] $t'=t'_1 t_2$. By induction hypothesis, $t_1 \rightarrow t'_1$ preserves the type. Therefore, $t' : T_{12}$.
	\item[(2)] $t'=v_1 t_2 \rightarrow v_1 t'_2$. Same as case (1). 
	\item[(3)] $t'=(\lambda x:T.t_1) v_2 \rightarrow \lbrack x \mapsto v_2 \rbrack t_2$. By Lemma~\ref{lemma:substitution}, substitution preserves typing. Therefore, $t' : T_{12}$.
\end{description}

\item[(\underline{T-IF})] $t=\textbf{if}\;t_1\;\textbf{then}\;t_2\;\textbf{else}\;t_3,\;\Gamma\vdash t_1:\text{bool},\;\Gamma\vdash t_2:S,\;\Gamma\vdash t_3:T,\;lub(S,T)\Rightarrow W,\;\Gamma\vdash t:W$. There are three rules by which $t'$ can be computed: E-IF-TRUE, E-IF-FALSE and E-IF.
\begin{description}
	\item[(1)] $t':t_2$. By rule T-IF, $lub(S,T) \Rightarrow W$ and by Lemma~\ref{lemma:lub}, $S <: W$, therefore $t':W$ by rule T-SUB.
	\item[(2)] $t':t_3$. Same as case (1).
	\item[(3)] $t':\textbf{if}\;t'_1\;\textbf{then}\;t_2\;\textbf{else}\;t_3$. By induction hypothesis, $t_1 \rightarrow t'_1$ preserves the type. Therefore, by rule T-IF, $t' : W$
\end{description}

\item[(\underline{T-ABS})] Vacuously fulfilled $\lambda(x:T_1).t_2$ is a value.  

\item[(\underline{T-VAR})] Cannot happen. 

\item[(\underline{T-TRUE})] Vacuously fulfilled since $t=\text{true}$ is a value. 

\item[(\underline{T-FALSE})] Vacuously fulfilled since $t=\text{false}$ is a value. 

\item[(\underline{T-SUB})] Result follows from induction hypothesis. 

\item[(\underline{T-NIL})] Vacuously fulfilled since $t=\textbf{nil}$ is a value.

\item[(\underline{T-CONS})] $t=\textbf{cons}\;t_1\;t_2,\;\Gamma\vdash t_1:C,\;,\Gamma\vdash t_2:D\;\text{list},\;\Gamma\vdash t:(C\sqcup D)\;\text{list}$. There are two rules by which $t'$ can be computed: E-CONS1 and E-CONS2. 
\begin{description}
	\item[(1)] $t'=\textbf{cons}\;t'_1\;t_2$. By induction hypothesis, $t_1 \rightarrow t'$ preserves the type. Therefore, by T-CONS, $t' : (C \sqcup D)\;\text{list}$.
	\item[(2)] $t'=\textbf{cons}\;v_1\;t'_2$. Same as case (1).
\end{description} 

\item[(\underline{T-NULL})] $t=\textbf{null}\;t_1, \;\Gamma \vdash t_1 : T\;\text{list},\;\Gamma\vdash t:\text{bool}$. By hypothesis, $t_1$ is either a value or it can take a step. If it can take a step, rule E-NULL applies. If its a value, by Lemma~\ref{lemma:canonical}, $t_1 = \text{nil}$ or $t_1 = (\textbf{cons}\;v_1\;...)$. If $t_1 = \textbf{nil}$, then rule E-NULL-TRUE applies. In case of  $t_1 = (\textbf{cons}\;v_1\;...)$, rule E-NULL-FALSE applies. 

\item[(\underline{T-HEAD})] $t=\textbf{head}\;t_1,\;\Gamma\vdash t_1 : T\;\text{list},\;\Gamma\vdash t:T$. There are two rules by which $t'$ can be computed: E-HEAD and E-HEADV.
\begin{description}
	\item[(1)] $t'=\textbf{head}\;t'_1$. By induction hypothesis, $t_1 \rightarrow t'_1$ preserves the type. Therefore, by T-HEAD, $t' : T$.
	\item[(2)] $t_1 = \textbf{cons}\;v_1\;v_2,\;\Gamma\vdash v_1 : T,\;t'=v_1$. Result is immediate, since $v_1 : T$.  
\end{description}

\item[(\underline{T-TAIL})] $t=\textbf{tail}\;t_1,\;\Gamma\vdash t_1 : T\;\text{list},\;\Gamma\vdash t:T\;\text{list}$. There are two rules by which $t'$ can be computed: E-TAIL and E-TAILV.
\begin{description}
	\item[(1)] $t'=\textbf{tail}\;t'_1$. By induction hypothesis, $t_1 \rightarrow t'_1$ preserves the type. Therefore, by T-TAIL, $t' : T$.
	\item[(2)] $t_1 = \textbf{cons}\;v_1\;v_2,\;\Gamma\vdash v_2 : T\;\text{list},\;t'=v_2$. Result is immediate, since $v_2 : T\;\text{list}$.  
\end{description}

\item[(\underline{T-QUERY})] $t=\textbf{query}\;C,\;\Gamma\vdash t:C\;\text{list}$. By applying rule E-QUERY, $t' = \textbf{cons}\;a_1\;...$. However, for each $a$, it is known that $\mathcal{K}\models a:C$, therefore $\{\;a\;\} <: C$ holds for each $a$ and $\{\;a_1\;\}\sqcup\;... <: C\;\text{list}$.

\item[(\underline{T-PROJ})] $t = t_1.R,\;\Gamma\vdash t_1 : C,\;\Gamma\vdash t:(\exists R^-.C)$. There are two rules by which $t'$ can be computed: E-PROJ and E-PROJV:
\begin{description}
	\item[(1)] $t'=t'_1.R$. By induction hypothesis, typing is preserved for $t_1$. Therefore, by T-PROJ, $t' : (\exists R^-.C)\;\text{list}$.
	\item[(2)] $t'=\sigma(\{\mathit{?X}\;\vert\;\mathit{?X}\in\mathcal{O}\wedge\mathcal{K} \models (a,\mathit{?X}):R\})=\textbf{cons}\;b_1\;...$. For $a$, it is known that $\mathcal{K} \models a:C$ and for each $b$ is known that $\mathcal{K} \models (a,b):R$ holds. Therefore, $\mathcal{K} \models b:(\exists R^-.C)$ must hold for each $b$. Thereby, $\{\;b_1\;\}\sqcup\;...<:(\exists R^-.C)$ and by S-LIST $(\{\;b_1\;\}\sqcup\;...)\;\text{list}<:(\exists R^-.C)\;\text{list}$
\end{description}
\item[(\underline{T-DISPATCH})] $\;$
\begin{lstlisting}[aboveskip=0pt,belowskip=0pt]
$t=\textbf{case}\;t_0\;\textbf{of}$
	$\textbf{type}\;C_1\;\textbf{as}\;x_1\;\textbf{->}\;t_1$
	$...$
	$\textbf{type}\;C_n\;\textbf{as}\;x_n\;\textbf{->}\;t_n$
	$\textbf{default}\;t_{n+1}$
$\Gamma\vdash t_0:D,\;\Gamma\vdash t_1:T_1,\;...,\;\Gamma\vdash t_n:T_n,\;\Gamma\vdash t_{n+1}:T_{n+1},$
$\overline{lub}(T_1,...,T_{n+1})\Rightarrow W,\;\Gamma\vdash t:W$
\end{lstlisting}
There are four rules by which $t'$ can be computed:  E-DISPATCH, E-DISPATCH-SUCC, E-DISPATCH-FAIL and E-DISPATCH-DEF. 
\begin{description}
	\item[(1)] $\;$
\begin{lstlisting}[aboveskip=0pt,belowskip=0pt]
$t'=\textbf{case}\;t'_0\;\textbf{of}$
	$\textbf{type}\;C_1\;\textbf{as}\;x_1\;\textbf{->}\;t_1$
	$...$
	$\textbf{type}\;C_n\;\textbf{as}\;x_n\;\textbf{->}\;t_n$
	$\textbf{default}\;t_{n+1}$
\end{lstlisting}
By induction hypothesis, $t_1 \rightarrow t'_1$ preserves the type. Therefore, by T-DISPATCH, $t':W$.
	\item[(2)] $t'=\lbrack x_1 \mapsto a\rbrack t_1,\;\Gamma\vdash t_1:T_1$. By Lemma~\ref{lemma:substitution}, substitution does not change the type of $t_1$. By Lemma~\ref{lemma:lub}, $T_1 <: W$ and therefore by rule T-SUB $t_1 : W$.
	\item[(3)] $\;$
\begin{lstlisting}[aboveskip=0pt,belowskip=0pt]
$t'=\textbf{case}\;a\;\textbf{of}$
	$\textbf{type}\;C_2\;\textbf{as}x_2\;\textbf{->}\;t_2$
	$...$
	$\textbf{type}\;C_n\;\textbf{as}x_n\;\textbf{->}\;t_n$
	$\textbf{default}\;t_{n+1}$
$\Gamma\vdash t_2:T_1,\;...,\;\Gamma\vdash t_n:T_n,\;\Gamma\vdash t_{n+1}:T_{n+1},$
$\overline{lub}(T_2,...,T_{n+1})\Rightarrow W',\;\Gamma\vdash t':W'$
\end{lstlisting}
The removal of the first case causes T-DISPATCH to assign type $t':W'$. Removal of $T_1$ makes $W'$ more specific then $W$, but $W' <: W$ holds. Therefore by, T-SUB $t':W$.
	\item[(4)] $t'=t_{n+1}\;\Gamma\vdash t_{n+1}:T_{n+1}$. By Lemma~\ref{lemma:lub}, $T_{n+1} <: W$, therefore by T-SUB $t' : W$. 
\end{description}

\item[(\underline{T-EQN})] $t_1\;\textbf{=}\;t_2,\;\Gamma\vdash t_1:C,\;\Gamma\vdash t_2:D,\;\Gamma\vdash t:\text{bool}$. There are $6$ different rules by which $t'$ can be computed: E-NOMINAL-TRUE, E-NOMINAL-FALSE, E-PRIM-TRUE, E-PRIM-FALSE, E-EQ1 and E-EQ2. 
\begin{description}
	\item[(1)] $t'=\text{true}$. Immediate by rule T-TRUE.
	\item[(2)] $t'=\text{false}$. Immediate by rule T-FALSE.
	\item[(3)] $t'=\text{true}$. Immediate by rule T-TRUE.
	\item[(4)] $t'=\text{false}$. Immediate by rule T-FALSE.
	\item[(5)] $t'=t'_1\textbf{=}t_2$. By induction hypothesis, $t_1 \rightarrow t'_1$. preserves the type. Therefore, by rule T-EQN, $t':\text{bool}$. 
	\item[(6)] $t'=v_1\textbf{=}t'_2$. By induction hypothesis, $t_2 \rightarrow t'_2$. preserves the type. Therefore, by rule T-EQN, $t':\text{bool}$. 
\end{description}

\item[(\underline{T-EQP})] $t_1\;\textbf{=}\;t_2,\;\Gamma\vdash t_1:\Pi_1,\;\Gamma\vdash t_2:\Pi_1$. Same as T-EQN. 

\item[(\underline{T-OBJ})] Vacuously fulfilled since $t=a$ is a value. 
\end{description}
\end{proof}

As a direct consequence of Theorems \ref{theorem:progress} and \ref{theorem:preservation}, a well-typed closed term does not get stuck during evaluation.


\section{Related work}
\label{sec:relatedWork}
$\lambda_{DL}$ is generally related to the integration of data models into programming languages. We consider four different ways of integrating such a data model: by using generic representations, by mappings into the target language, through a preprocessing step before compilation, or through language extensions or custom languages. 

\paragraph*{Generic representations}
Generic representations offer easy integration into programming
languages and have the advantage that they can represent anything the
data can model, e.g., generic representations (such as
DOM\footnote{\url{https://www.w3.org/DOM/}}) for
XML~\cite{DBLP:conf/icfp/WallaceR99}. This approach has also been
applied to semantic data. Representations can vary, however the most
popular ones include axiom-based approaches (e.g.,
\cite{DBLP:journals/semweb/HorridgeB11}), graph-based ones (e.g.,
\cite{DBLP:conf/www/CarrollDDRSW04}) or statement-based ones (e.g.,
RDF4J\footnote{\url{http://rdf4j.org/}}). All these approaches are
error-prone in so far that code on the generic representations is
not type-checked in terms of the involved conceptualizations.   

\paragraph*{Mappings}
Mapping approaches on the other hand use schematic information of the
data model to create types in the target language. Type checking can
be used thus to check the valid use of the derived types in
programs. This approach has been successfully used for
SQL~\cite{DBLP:conf/sigmod/ONeil08},
XML~\cite{DBLP:conf/icfp/WallaceR99,DBLP:conf/ssdgp/LammelM06,DBLP:conf/icoodb/AlagicB09},
and more
generally~\cite{DBLP:conf/gttse/LammelM06,DBLP:conf/popl/SymeBTMP12}. Naturally,
mappings have been studied in a semantic data context, too. The focus
is on transforming conceptual statements into types of the programming
language. Frameworks include
ActiveRDF~\cite{Oren:2008:AES:1412759.1412997},
Alibaba\footnote{\url{https://bitbucket.org/openrdf/alibaba}},
Owl2Java~\cite{DBLP:conf/seke/KalyanpurPBP04},
Jastor\footnote{\url{http://jastor.sourceforge.net/}},
RDFReactor\footnote{\url{http://semanticweb.org/wiki/RDFReactor}},
OntologyBeanGenerator\footnote{\url{http://protege.cim3.net/cgi-bin/wiki.pl?OntologyBeanGenerator}},
\`{A}gogo~\cite{AGOGO} and
LITEQ~\cite{DBLP:conf/semweb/LeinbergerSLSTV14}. However, mapping
approaches are problematic for semantic data. For one, the
transformation of
statements such as those shown in line 1 of Listing \ref{ax:example}
is not trivial due to the mixture of nominal and structural
typing. Extremely general information on domains and ranges of roles
such as $\text{\tt influencedBy}$ occurs frequently. The question
arises what types support such a role. Frameworks usually resolve the
situation by assigning the role to every type they create. In terms of
the codomain of the role, they usually assign the most general
available type and leave it to the developer to cast the values to
their correct types---this is an error-prone approach. Lastly, all mapping frameworks have problems with the high number of potential types in semantic data sources.   

\paragraph*{Precompilation}
A separate precompilation step, where the source code is statically
analyzed beforehand for DSL usage and then verified or transformed is
another way to solve the problem of integrating data models into
programming languages. Especially queries embedded in programming
languages can be verified in this manner. This approach has been
applied to, for example, SQL
queries~\cite{DBLP:journals/tosem/WassermannGSD07}. The approach has
been applied to semantic data in a limited manner~\cite{SWOBE}---for
queries that can be typed with primitive types such as integer.

\paragraph*{Language extensions and custom type systems}
The most powerful approaches extend existing languages or create new
type systems to accommodate the specific requirements of the data
model. Examples for such extensions are concerned with relationships
between objects~\cite{Bierman2005} and easy data access to relational
and XML data~\cite{export:77415}. 
Another example concerns programming language support for the XML data
model specifically in terms of regular expression type, as in the
languages CDuce~\cite{cduce} and XDuce~\cite{xduce}. While semantic
data can be seen as somewhat semi-structured and is often serialized
in XML, the XML-focused approaches do not address the logics-based
challenges regarding semantic data. Another related approach is the idea of functional logic programming~\cite{Hanus94JLP}. However, $\lambda_{DL}$ emphasizes type-checking on data axiomatized in logic over the integration of the logic programming paradigm into a language. Given its typecase constructs, $\lambda_{DL}$ is also related to other forms of  typecases~\cite{DBLP:journals/jfp/AbadiCPR95,DBLP:journals/jfp/CraryWM02,DBLP:conf/tldi/LammelJ03}. However, since semantic data cannot be adequately expressed with existing typing mechanics, these approaches cannot fully solve the problems. 

Language extensions and custom approaches have also been implemented
for semantic data. In one approach~\cite{DBLP:conf/esws/PaarV11}, the
C\# compiler was extended to allow for OWL and XSD types in C\#. The
main technical difference to $\lambda_{DL}$ is that $\lambda_{DL}$
makes use of the knowledge system for typing and subtyping
judgments. $\lambda_{DL}$ can therefore make use of inferred data and
has a strong typing mechanism. There is also work on custom languages
that use static type-checking for querying and light scripting in
order to avoid runtime
errors~\cite{DBLP:conf/ershov/CiobanuHS14,DBLP:journals/jlp/CiobanuHS15}. However,
the types are again limited in these cases, as they only consider explicitly given statements. Furthermore, they face the same difficulties as mapping approaches when it comes to schema information --- they rely on domain and range specifications for predicates to assign types. 


\section{Discussion and future work}
\label{sec:conclusion}

In this paper, we have motivated, introduced and studied a type system for semantic data that is built around concept expressions as types as well as queries in a simple $\lambda$-calculus. We have shown that by using conceptualizations as they are defined in the knowledge system itself, type safety can be achieved. This helps in writing less error-prone programs, even when facing knowledge systems that evolve. However, the work can be extended in several ways. 

\paragraph*{Gradual typing} 
A byproduct of achieving type safety are the rather hard restrictions
by the schema. This can be seen best for the $\text{\tt influencedBy}$
role as described in the example. The knowledge system could not prove
that $\text{\tt MusicArtist} \sqsubseteq \exists \text{\tt
  influencedBy}.\top$, therefore the influences of music artists could
not be computed directly. In the example, this was a correct choice as
not every music artist was influenced by something. But there are also
scenarios where it is reasonable to assume that, for a specific data
source, this will be the case even though the schema does not
explicitly state so, simply because schemata for semantic data strive
to be applicable to different and evolving data sources. Also, the
semantic data source may have been created by conversion from more constrained
data, e.g., in a SQL database. One way to include such background knowledge of a developer would be to adopt ideas originating from gradual typing~\cite{Siek06gradualtyping} for $\lambda_{DL}$. A `lenient' $\lambda$ could be introduced, which accepts values even though the subtyping relation cannot be proven. However, even in this case, one would still check if the intersection of the functions domain and of the value applied to the function is non-empty in order to avoid grave mistakes. 

\paragraph*{$\mathbf{\lambda_{DL}}$ and System F}
So far, we have only considered a simply typed $\lambda$-calculus for
the integration of semantic data into a functional language. However,
programming languages typically feature polymorphic definitions, e.g.,
for list-processing function combinators. A comprehensive integration 
of description logics and polymorphism (with
$\text{System}\;F_{<:}$~\cite{DBLP:conf/ifip/Reynolds83} as starting
point) including aspects of subtyping is not straightforward.


\paragraph*{Modification of the semantic data}
Of course, it is also desirable that semantic data can be modified by
an extended $\lambda_{DL}$. However, due to facts inferred by the knowledge system, this is non-trivial. Given the facts about music artists in Listing \ref{ax:example} and the goal to remove the (implicit) fact that the $\text{\tt beatles}$ have made a song. This cannot be removed directly. Instead, either the fact that the $\text{\tt beatles}$ are of type $\text{\tt MusicArtist}$ or the fact that they have been played by $\text{\tt coolFm}$ must be removed. In order to integrate modification of knowledge systems into $\lambda_{DL}$, the theory of knowledge revision based on AGM theory~\cite{Qi2006} must be considered and properly integrated into the language. 

\paragraph*{Enhanced querying}
Another area of future work concerns the query system. Queries, as
they are currently implemented, are limited in their expressive
power. A simple extension are queries for roles, such as
$\text{\tt influencedBy}$ that result in sets of pairs. Typing such
queries is possible via the addition of tuples to $\lambda_{DL}$. The
addition of query languages closer to the power of SQL is also
possible. The biggest challenge in this regard is query
subsumption. When such queries are typed in the programming language,
subsumption checks are necessary to determine whether a function can
be applied to query results. Therefore only query languages with
decidable query subsumption are to be considered (e.g.,
\cite{DBLP:conf/ijcai/BourhisKR15}).




\bibliographystyle{abbrvnat}
\bibliography{references}



\newpage
\appendix
\section{Appendix}


\subsection{Remaining reduction rules}

\begin{figure}[h!]
\centering
\begin{boxedminipage}{0.5\textwidth}
\vspace{2mm}
	\ax{\text{E-LETV}}
		{\<let> x=v_1 \|in| t_2 \rightarrow \lbrack x \mapsto v_1\rbrack t_2}

	\ir{\text{E-LET}}
		{t_1 \rightarrow t'_1}
		{\<let> x=t_1 \|in| t_2 \rightarrow \<let> x=t'_1 \|in| t_2}

	\ax{E-FIXV}
		{\<fix> (\lambda x:T_1.t_2) \rightarrow \lbrack x \mapsto (\<fix> (\lambda x:T_1.t_2)) \rbrack t_2}

	\ir{E-FIX}
		{t_1 \rightarrow t'_1}
		{\<fix> t_1 \rightarrow  \<fix> t'_1}

	\ir{\text{E-APP1}}
		{t_1 \rightarrow t'_1}
		{t_1 t_2 \rightarrow t'_1 t_2}

	\ir{\text{E-APP2}}
		{t_2 \rightarrow t'_2}
		{v_1 t_2 \rightarrow v_1 t'_2}

	\ax{\text{E-APPABS}}
		{(\lambda x:T.t_1)\;v_2 \rightarrow \lbrack x \mapsto v_2 \rbrack t_1}

	\ax{\text{E-IF-TRUE}}{\<if> \text{true} \|then| t_2 \|else| t_3 \rightarrow t_2}
	\ax{\text{E-IF-FALSE}}{\<if> \text{false}\|then| t_2 \|else| t_3 \rightarrow t_3}

	\ir{\text{E-IF}}
		{t_1 \rightarrow t'_1}
		{\<if> t_1 \|then| t_2 \|else| t_3 \rightarrow \<if> t'_1 \|then| t_2 \|else| t_3}
\end{boxedminipage}
\caption{Reduction rules for constructs unrelated to KB.}
\label{fig:normalSemantics}
\end{figure}

\begin{figure}[h!]
\centering
\begin{boxedminipage}{0.5\textwidth}
	\ir{\text{E-CONS1}}
		{t_1 \rightarrow t'_1}
		{\textbf{cons}\; t_1\; t_2 \rightarrow \textbf{cons}\; t'_1\; t_2}

	\ir{\text{E-CONS2}}
		{t_2 \rightarrow t'_2}
		{\textbf{cons}\; v_1\; t_2 \rightarrow \textbf{cons}\; v_1\; t'_2}

	\ax{\text{E-NULL-TRUE}}
		{\<null> \textbf{nil} \rightarrow \text{true}}

	\ax{\text{E-NULL-FALSE}}
		{\<null> \textbf{cons}\;v_1\;v_2 \rightarrow \text{false}}

	\ir{\text{E-NULL}}
		{t_1 \rightarrow t'_1}
		{\<null> t_1 \rightarrow \<null> t'_1}

	\ax{\text{E-HEADV}}
		{\<head> \textbf{cons}\;v_1\;v_2 \rightarrow v_1}

	\ir{\text{E-HEAD}}
		{t_1 \rightarrow t'_1}
		{\<head> t_1 \rightarrow \<head> t'_1}

	\ax{\text{E-TAILV}}
		{\<tail> \textbf{cons}\;v_1\;v_2 \rightarrow v_2}

	\ir{\text{E-TAIL}}
		{t_1 \rightarrow t'_1}
		{\<tail> t_1 \rightarrow \<tail> t'_1}
\end{boxedminipage}
\caption{Reduction rules for lists.}
\label{fig:listSemantics}
\end{figure}

\newpage
\subsection{Greatest-lower bound}
\begin{figure}[h]
\centering
\begin{boxedminipage}{0.45\textwidth}
	\ax{\text{GLB-PRIMITIVE}}
		{glb(\pi_1,\pi_1) \Rightarrow \pi_1}

	\ax{\text{GLB-CONCEPT}}
		{glb(C,D)\Rightarrow C \sqcap D}

	\ir{\text{GLB-LIST}}
		{glb(S,T) \Rightarrow W}
		{glb(S\;\text{list},T\;\text{list}) \Rightarrow W\;\text{list}}

	\ir{\text{GLB-FUNC}}
		{lub(S_1,T_1)\Rightarrow W_1 \dnp glb(S_2, T_2) \Rightarrow W_2}
		{glb(S_1 \rightarrow S_2, T_1 \rightarrow T_2) \Rightarrow W_1 \rightarrow W_2}		
\end{boxedminipage}
\caption{Greatest lower bound of types.}
\label{fig:glb}
\end{figure}

\subsection{Prototypical implementation}

A prototypical implementation, showing the feasibility of $\lambda_{DL}$ is available at \url{http://west.uni-koblenz.de/de/lambda-dl}. The interpreter itself is written in F\# while relying on a Java-based HermiT reasoner for inferencing. Most of the interpreter is based on the approach shown by \cite{Pierce:2002:TPL:509043}. An important difference (besides the actual rules) is that evaluation and typing functions take a knowledge base as an additional parameter. In case of HermiT knowledge bases, a wrapper is passed to those functions. The wrapper serializes queries issued by the evaluation and typing functions and calls a Java program, which then in turn deserializes the queries and calls the reasoner. 


\end{document}